\documentclass[letter,11pt]{article}

\usepackage{fullpage}
\usepackage{hyperref}
\usepackage[numbers]{natbib}
\usepackage{xspace}
\usepackage{xcolor}
\usepackage{graphicx}
\usepackage{amsmath,amssymb,amsthm}
\usepackage{tikz}
\usetikzlibrary{chains,arrows}

\usepackage{color}            

\theoremstyle{plain}
\newtheorem{proposition}{Proposition}
\newtheorem{theorem}{Theorem}
\newtheorem{lemma}{Lemma}
\newtheorem{corollary}{Corollary}
\theoremstyle{definition}

\newtheorem{thm}{Theorem}
\newenvironment{thmbis}[1]
  {\renewcommand{\thethm}{\ref{#1}$'$}%
   \addtocounter{thm}{-1}%
   \begin{thm}}
  {\end{thm}}

\newcommand{\Ex}[1]{\mbox{\rm\bf E}\left[#1\right]}

\begin{document}

\title{Algorithms as Mechanisms:\\ The Price of Anarchy of Relax-and-Round}

\author{%
	Paul D\"{u}tting\thanks{Department of Computer Science, ETH Z\"urich, Universit\"atstrasse 6, 8092 Z\"urich, Switzerland, Email: \texttt{paul.duetting@inf.ethz.ch}.}
	\and
	Thomas Kesselheim\thanks{Max-Planck-Institut f\"ur Informatik, Campus E1 4, 66123 Saarbr\"ucken, Germany, Email: \texttt{thomas.kesselheim@mpi-inf.mpg.de}.}
	\and
	\'{E}va Tardos\thanks{Department of Computer Science, Cornell University, Gates Hall, Ithaca, NY 14853, USA, Email: \texttt{eva@cs.cornell.edu}.}
}

\maketitle

\begin{abstract}
Many algorithms that are originally designed without explicitly considering incentive properties are later combined with simple pricing rules and used as mechanisms.  The resulting mechanisms are often natural and simple to understand. But how good are these algorithms as mechanisms? Truthful reporting of valuations is typically not a dominant strategy (certainly not with a pay-your-bid, first-price rule, but it is likely not a good strategy even with a critical value, or second-price style rule either).  Our goal is to show that a wide class of approximation algorithms yields this way mechanisms with low Price of Anarchy.

The seminal result of Lucier and Borodin \cite{LucierBorodin10} shows that combining a greedy algorithm that is an $\alpha$-approximation algorithm
with a pay-your-bid payment rule yields a mechanism whose Price of Anarchy is $O(\alpha)$. In this paper we significantly extend the class of algorithms for which such a result is available by showing that this close connection between approximation ratio on the one hand and Price of Anarchy on the other also holds for the design principle of {\em relaxation and rounding} provided that the relaxation is {\em smooth} and the rounding is {\em oblivious}.

We demonstrate the far-reaching consequences of our result by showing its implications for sparse packing integer programs, such as multi-unit auctions and generalized matching, for the maximum traveling salesman problem, for combinatorial auctions, and for single source unsplittable flow problems. In all these problems our approach leads to novel simple, near-optimal mechanisms whose Price of Anarchy either matches or beats the performance guarantees of known mechanisms.
\end{abstract}

\section{Introduction}

Mechanism design---or ``reverse'' game theory---is concerned with protocols, or mechanisms, through which potentially selfish agents interact with one another. The basic assumption is that the data is held by the agents, who may behave strategically. The goal is then to achieve outcomes that approximate the social optimum in a wide range of strategic equilibria.

The most sweeping positive result that one could hope for in this context---with some professional bias of course---is a general reduction from mechanism design to algorithm design, showing that mechanism design is just as ``easy'' as algorithm design. Specifically, one could hope that using algorithms as they are and charging bidders their respective bids yields mechanisms whose equilibria are close to optimal.

Why would this be appealing? Such a result would make the entire toolbox of algorithm design available to mechanism design, significantly broadening the tools currently available. It would also ``hide'' the incentives aspect from the designer, who would then no longer need to worry about possible manipulations through the agents. He could simply focus on the problem of computing near optimal solutions for the claimed
input. Finally, the resulting mechanisms would enjoy a simplicity well beyond that found in most state-of-the art mechanisms.

Our goal in this paper is to identify general algorithm design principles that work well when used {\em as} mechanisms. We cannot expect this to be the case for all algorithms. Identifying algorithm design principles that automatically work well as mechanisms would, in some sense, give us the vocabulary to which we---as algorithm designers---should confine ourselves if we expect that our algorithm will be used in strategic environments. Specifically, it would equip us with the tools to design {\em simple} and {\em robust} mechanisms for such settings.

Borodin and Lucier \cite{LucierBorodin10} showed that greedy algorithms have this property: Any equilibrium of a greedy algorithm that is an $\alpha$-approximation algorithm is within $O(\alpha)$ of the optimal solution. Our main result is to show that the common design principle of {\em relaxation and rounding} also preserves the approximation guarantee as Price of Anarchy guarantee provided that the relaxation is {\em smooth} and the rounding process is {\em oblivious} (more on this below). The canonical example of this approach are integer linear programs that are relaxed to a fractional domain, then the relaxation is  solved to optimality and converted into an integer solution via randomized rounding. Other examples that follow this pattern come from semi-definite programming or involve relaxing one combinatorial problem to another.

Our result has---as we show---far-reaching consequences in mechanism design: It leads to novel simple, yet near-optimal mechanisms for sparse packing integer programs, such as multi-unit auctions and generalized matching, for the maximum traveling salesman problem, for combinatorial auctions and for single source routing problems. In all cases we obtain Price of Anarchy bounds that match or beat known Price of Anarchy guarantees, or they are the first non-trivial guarantees for the respective problem.

\subsection{Our Contributions}

Our results concern the algorithmic blueprint of {\em relaxation and rounding} (see, e.g., \cite{Vazirani2001}). In this approach a problem $\Pi$ is {\em relaxed} to a problem $\Pi'$, with the purpose of rendering exact optimization computationally tractable. Having found the optimal relaxed solution $x'$, another algorithm derives a solution $x$ to the original problem. This process is called {\em rounding}.

Many rounding schemes in text books as well as highly sophisticated ones are \emph{oblivious}. That is, they do {\em not} require knowledge of the objective function. Up to this point, to the best of our knowledge, this property---though wide-spread---has never proven useful. In this paper, we show that oblivious rounding schemes preserve bounds on the Price of Anarchy. That is, applying an $\alpha$-approximate oblivious rounding scheme on a problem with a Price of Anarchy bound $\beta$, the combined mechanism has Price of Anarchy at most $O(\alpha \beta)$.

We thus translate the relax-and-round approach from algorithm design into mechanism design: If we relax a problem into a problem with Price of Anarchy $\beta$ and round the solution to the relaxed problem with an $\alpha$-approximate oblivious rounding scheme, the resulting mechanism has a Price of Anarchy of $O(\alpha \beta)$. The bound does not only apply to Nash equilibria, but also extends to the Bayesian setting as well as to learning outcomes (coarse correlated equilibria). See Section \ref{sec:discussion} for discussion on the existence and computational complexity of finding such outcomes.

\subsubsection{Main Result} 
Our main result leverages the power of the smoothness framework of Roughgarden \cite{Roughgarden09,Roughgarden12} and Syrgkanis and Tardos \cite{SyrgkanisT13}.

At the heart of this framework is the notion of a $(\lambda,\mu)$-smooth mechanism, where $\lambda, \mu \ge 0$. The main result is that a mechanism that is $(\lambda,\mu)$-smooth achieves a Price of Anarchy of $\beta(\lambda,\mu) = \max(1,\mu)/\lambda$ with respect to a broad range of equilibrium concepts including learning outcomes. Furthermore, the simultaneous and sequential composition of $(\lambda,\mu)$-smooth mechanisms is again $(\lambda,\mu)$-smooth. Ideally, $\lambda = 1$ and $\mu \le 1$ in which case this result tells us that all equilibria of the mechanism are socially optimal; otherwise, if $\lambda < 1$ or $\mu > 1$, then this result tells us which fraction of the optimal social welfare the mechanism is guaranteed to get at any equilibrium.

The other crucial ingredient to our main result is the notion of an $\alpha$-{\em approximate oblivious rounding scheme}, where $\alpha \ge 1$. This is a (possibly randomized) rounding scheme for translating a solution $x'$ to the relaxed problem $\Pi'$ into a solution $x$ to the original problem $\Pi$ so that for all possible valuation profiles each agent is guaranteed to get, in expectation, a $1/\alpha$-fraction of the value that it would have had for the solution to the relaxed problem.

Clearly an $\alpha$-approximate oblivious rounding scheme, when combined with optimally solving the relaxed problem, leads to an approximation ratio of $\alpha$. We show that it also approximately preserves the Price of Anarchy of the relaxation. We focus on pay-your-bid mechanisms for concreteness. Our result actually applies to a broad range of mechanisms and can also be extended to include settings where the relaxation is not solved optimally; we discuss these extensions in Section \ref{sec:extensions}.

\begin{theorem}[Main Theorem, informal]
Consider problem $\Pi$ and a relaxation $\Pi'.$ Suppose the pay-your-bid mechanism $M$ for $\Pi$ is derived from the pay-your-bid mechanism $M'$ for $\Pi'$. If $M'$ is $(\lambda,\mu)$-smooth, then $M$ is $(\lambda/(2\alpha),\mu)$-smooth.
\end{theorem}

\begin{corollary}
The Price of Anarchy established via smoothness of mechanism $M'$ of $\beta$ translates into a smooth Price of Anarchy bound for mechanism $M$ of $2\alpha\beta$ extending to both Bayesian Nash equilibria and learning outcomes.
\end{corollary}

Our main theorem can be strengthened if the relaxation satisfies a slightly stronger smoothness condition, also parametrized by $\lambda$ and $\mu$, which all our application do. In this case we can show that the derived mechanism is $(\lambda/\alpha,\mu)$-smooth; and the corollary would read ``a Price of Anarchy of $\beta$ translates into a Price of Anarchy of $\alpha\beta$.''

\subsubsection{Applications}
We demonstrate the far-reaching consequences of our result by applying it to a broad range of optimization problems. For each of these problems we show the existence of a smooth relaxation and the existence of an oblivious rounding scheme. We note that in all of our applications, it is important to use the relaxation to show smoothness of the problem. For example, optimally solving the original (integer) problem would give a very high Price of Anarchy.

\paragraph{Sparse Packing Integer Programs.}
The first problem we consider are multi-unit auctions with $n$ bidders and $m$ items, where bidders have unconstrained valuations. The underlying optimization problem has a natural LP relaxation, which we show is $(1/2,1)$-smooth. Using the $8$-approximate oblivious rounding scheme of Bansal et al.~\cite{BansalKNS10}, our framework yields a constant PoA. This is quite remarkable as solving the integral optimization problem leads to a PoA that grows linearly in $n$ and $m$.

We then consider the generalized assignment problem in which $n$ bidders have unit-demand valuations for a certain amount of one of $k$ services and allocations of services to bidders must respect the limited availability of each service. For this problem we also show $(1/2,1)$-smoothness, and use the $8$-approximate oblivious rounding scheme of \cite{BansalKNS10} to obtain a constant PoA.

Both these results are in fact special cases of a general result regarding sparse packing integer programs (PIP) that we show. Namely, the pay-your-bid mechanism that solves the canonical relaxation of a PIP with column sparsity $d$ is $(1/2,d+1)$-smooth. Multi-unit auctions and the generalized assignment problem have $d =1$; combinatorial auctions in which each bidder is interested in at most $d$ items simultaneously have $d \ge 1$. For general PIPs the rounding scheme of \cite{BansalKNS10} is $O(d)$-approximate. We get a PoA of $O(d^2)$.

\paragraph{Maximum Traveling Salesman.}
Our second application is the maximization variant of the classic traveling salesman problem (max-TSP). We think of the problem as a game where each edge $e$ has a value for being included, and the goal of the mechanism is to select a TSP of maximum total value.
The classic algorithm for this problem is a $2$-approximation \cite{Fisher1979}. It proceeds by computing a cycle cover, dropping an edge from each cycle, and connecting the resulting paths in an arbitrary manner to obtain a solution.
We prove this can be thought off as a $2$-approximate oblivious rounding scheme and show, through a novel combinatorial argument, that the relaxation is $(1/2,3)$-smooth. We thus obtain a Price of Anarchy of $12$.

The best approximation guarantee for max-TSP is a $3/2$-approximation due to Kaplan et al.~\cite{KaplanLSS03}. The same approximation ratio is achieved by a (much simpler) algorithm of Paluch et al.~\cite{PaluchEZ12}. We show that this algorithm---just as the basic algorithm---can be interpreted as a relax-and-round algorithm. Generalizing the arguments for the basic algorithm to the (different) relaxation used in this interpretation, we show that this algorithm achieves a Price of Anarchy that is by a factor $3/4$ better than the Price of Anarchy of the basic algorithm.

These examples are especially interesting as they show how a seemingly combinatorial algorithm can be re-stated within our framework. They also represent the first non-trivial PoA bounds for this problem.

\paragraph{Combinatorial Auctions.}
We also consider the ``canonical'' mechanism design problem of combinatorial auctions in which valuations are restricted to come from a certain class. Our first result concerns fractionally subadditive, or XOS, valuations \cite{LehmannLN06}. We show that the pay-your-bid mechanism for the canonical LP relaxation is $(1/2,1)$-smooth. Using Feige's ingenious $e/(e-1)$-approximate oblivious rounding scheme \cite{Feige09}, our main result implies an upper bound on the Price of Anarchy of $4e/(e-1)$.

We then show how to extend this result to the recently proposed hierarchy of $\mathcal{MPH}$-$k$ valuations \cite{FeigeFIILS14}. Levels of the hierarchy correspond to the degree of complementarity in a given function. The lowest level $k=1$ coincides with the class of XOS/fractionally subadditive valuations; the highest level $k = m$ can be shown to comprise all monotone valuation functions. We show that for $\mathcal{MPH}$-$k$ valuations the LP relaxation is $(1/2,k+1)$ smooth. Together  with the $O(k)$-approximate oblivious rounding scheme of \cite{FeigeFIILS14} we obtain a Price of Anarchy of $\Theta(k^2)$.

These results nicely complement the recent work on ``simple auctions'' such as \cite{ChristodoulouKS08,BhawalkarR11,FeldmanFGL13,DuHeSt13,Roughgarden14}, answering an open question of Babaioff et al.~\cite{BabaioffLNL14} regarding the Price of Anarchy of {\em direct} mechanism based on approximation algorithms in these settings. 
The advantage of having a direct mechanism for this problem is that one can consider simple bidding strategies (such as bidding half the value) to establish the performance guarantees, whereas in indirect mechanisms such as combinatorial auctions with item bidding the computational effort is effectively shifted to the bidders.

\paragraph{Single Source Unsplittable Flow.}
The final problem that we consider are multi-commodity flow (MCF) problems with a single source (or  target). In these problems we are given a capacitated, directed network and a set of requests consisting of a target (a source) and a demand, corresponding to requests of, say different information, held at the source. The goal is to maximize the total demand routed (or the total value of the demand routed), subject to feasibility. We assume each player has a demand for some flow to be routed from a shared source to a terminal specific to the player, and the player has a private value for routing this flow.

For this problem we show that the natural LP relaxation is $(1/2,1)$-smooth. A $(1+\epsilon)$-approximate oblivious rounding scheme for high enough capacities is obtained through an adaptation of the ``original'' randomized rounding algorithm of \cite{Raghavan1988,RaghavanThompson1987}. This yields a PoA of $2(1+\epsilon).$

An interesting feature of this result is that the LP can be solved greedily through a variant of Ford-Fulkerson which allows us to exploit the known connection to smoothness \cite{LucierBorodin10,SyrgkanisT13}. Crucially, the reference to these results has to be on the fractional level, as a greedy procedure on the integral level achieves a significantly worse approximation guarantee.

\subsection{Related Work}

Our work is closely related to the literature on so-called ``back-box reductions'', which has led to some of the most impressive results in algorithmic mechanism design (such as \cite{LaviS05,BriestKV05,DughmiR14,DughmiRY11,BabaioffKS10,BabaioffKS13}). This approach takes an algorithm, and aims to implement the algorithm's outcome via a game. To this end it typically modifies the algorithm and adds a sophisticated payment scheme. Our approach is different in that we consider an algorithm without any modification, introduce a simple payment rule, such as the ``pay your bid'' rule, and understand the expected outcomes of the resulting game.

Lavi and Swamy \cite{LaviS05} use {\em randomized meta rounding} \cite{CarrVempala02} to turn LP-based approximation algorithms for packing domains into truthful-in-expectation mechanisms. Our result is similar in spirit as it demonstrates the implications of {\em obliviousness} for non-truthful mechanism design. The property that we need, however, is less stringent and shared by most rounding algorithms. Another important difference is that our approach is not limited to packing domains.

Briest et al.~\cite{BriestKV05} show how pseudo-polynomial approximation algorithms for single-parameter problems can be turned into a truthful fully polynomial-time approximation schemes (FPTAS). Dughmi et al.~\cite{DughmiR14} prove that every welfare-maximization problem that admits a FPTAS and can be encoded as a packing problem also admits a truthful-in-expectation randomized mechanism that is an FPTAS. Unlike our approach these approaches are limited to single-parameter problems, or to multi-parameter problems with packing structure.

Dughmi et al.~\cite{DughmiRY11} present a general framework that also looks at the fractional relaxation of the problem. They show that if the rounding procedure has a certain property, which they refer to as {\em convex rounding}, then the resulting algorithm is truthful. They instantiate this framework to design a truthful-in-expectation mechanism for CAs with matroid-rank-sum valuations (which are strictly less general than submodular). The main difference to our work is that standard rounding procedures are often oblivious but typically not convex.

Babioff et al.~\cite{BabaioffKS10,BabaioffKS13} show how to transform a (cycle-)monotone algorithm into a truthful-in-expectation mechanism using a {\em single call} to the algorithm. The resulting mechanism coincides with the algorithm with high probability. This work differs from ours in that it only applies to monotone or cycle-monotone algorithms.

By insisting on truthfulness, or truthfulness-in-expectation, as a solution concept, all these approaches face certain natural barriers to how good they can get (see, e.g., \cite{PapadimitriouSS08,ChawlaIL12}). In addition, they typically do not lead to simple, practical mechanisms. For example, despite running times technically being polynomial, these mechanisms require far more computational effort than standard approximation algorithms for the underlying optimization problem. In some cases, the reduction yields mechanisms in which the approximation guarantee is tight on every single instance (not only in the worst case). That is, even when the optimization problem is trivial, the mechanism sacrifices the solution quality for incentives.

\section{Preliminaries}
\label{sec:prelims}

\paragraph{Algorithm Design Basics.} We consider maximization problems $\Pi$ in which the goal is to determine a feasible outcome $x \in \Omega$ that maximizes total weight given by $w(x)$ for non-negative a weight function $w\colon \Omega \rightarrow \mathbb{R}_{\ge 0}$. A potentially randomized algorithm $A$ receives the functions $w$ as input and computes an output $A(w) \in \Omega$. The algorithm is an $\alpha$-{\em approximation algorithm}, for $\alpha \ge 1$, if for all weights $w$, $\Ex{w(A(w))} \ge \frac{1}{\alpha} \cdot \max_{x \in \Omega} w(x)$.

We are interested in {\em relax-and-round algorithms}. These algorithms first relax the problem $\Pi$ to $\Pi'$ by extending the space of feasible outcomes to $\Omega' \supseteq \Omega$ and generalizing weight functions $w$ to all $x \in \Omega'$. They compute an optimal solution $x' \in \Omega'$ to the relaxed problem. Then a solution $x \in \Omega$ of the original problem is derived based on $x' \in \Omega'$, typically via randomized rounding.

A rounding algorithm is {\em oblivious} if it does not require knowledge of the actual objective function $w$, beyond the fact that $x'$ was optimized with respect to $w$. Formally, a rounding scheme is an $\alpha$-{\em approximate oblivious rounding scheme} if, given some relaxed solution $x'$, it computes a solution $x$ such that for all $w$, $\Ex{w(x)} \ge \frac{1}{\alpha} w(x').$ Clearly, a relax-and-round algorithm based on an $\alpha$-approximate oblivious rounding scheme is an $\alpha$-approximation algorithm.

\paragraph{Mechanism Design Basics.}
Our results apply to general multi-parameter mechanism design problems $\Pi$ in which agents $N = \{1,\dots,n\}$ interact to select an element from a set $\Omega$ of outcomes.
Each agent has a valuation function $v_i\colon \Omega \rightarrow \mathbb{R}_{\ge 0}.$ We use $v$ for the valuation profile that specifies a valuation for each agent, and $v_{-i}$ to denote the valuations of the agents other than $i$. The quality of an outcome $x \in \Omega$ is measured in terms of its social welfare $\sum_{i \in N} v_i(x).$

We consider direct mechanisms $M$ that ask the agents to report their valuations. We refer to the reported valuations as bids and denote them by $b$.
The mechanism uses outcome rule $f$ to compute an outcome $f(b) \in \Omega$ and payment rule $p$ to compute payments $p(b) \in \mathbb{R}_{\ge 0}$. Both the computation of the outcome and the payments can be randomized.
We are specifically interested in {\em pay-your-bid mechanisms}, in which agents are asked to pay what they have bid on the outcome they get. In other words, in a pay-your-bid mechanism $M = (f,p)$, $p_i(b) = b_i(f_i(b)).$ 
We assume that the agents have quasi-linear utilities and that they are risk neutral. That is, we assume that agent $i$'s expected utility in mechanism $M=(f,p)$ is given by $u_i(b,v_i) = \Ex{v_i(f(b))} - \Ex{p_i(b)}.$

For the game-theoretic analysis we distinguish two settings. In the {\em complete information} setting agents know each others' valuations, and a potentially randomized bid profile $b$ that may depend on $v$ is a {\em mixed Nash equilibrium} if for all agents $i \in N$ and possible deviations $b'_i$ that may depend on $v$, $\mathbf{E}_{b}[u_i(b,v_i)] \ge \mathbf{E}_{b'_i,b_{-i}}[u_i((b'_i,b_{-i}),v_i)].$ In the {\em incomplete information} setting valuations are drawn from independent distributions $D_i$, and each agent $i \in N$ knows its own valuation $v_i$ and the distributions $D_{-i}$ from which the other agents valuations are drawn. A {\em mixed Bayes-Nash equilibrium} is a potentially randomized bid profile $b_i$ that may depend on this agent's valuation $v_i$ and the distributions $D_{-i}$ from which the other agents' valuations are drawn such that for all agents $i \in N$ and potential deviations $b'_i$ which are also allowed to depend on $v_i$ and $D_{-i}$, $\mathbf{E}_{b,v_{-i}}[u_i(b,v_i)] \ge \mathbf{E}_{b'_i,b_{-i},v_{-i}}[u_i((b'_i,b_{-i}),v_i)].$

\paragraph{Price of Anarchy.}
We evaluate the quality of mechanisms by their {\em Price of Anarchy}.
The Price of Anarchy with respect to Nash equilibria (PoA) is the worst ratio between the optimal social welfare and the expected welfare in a mixed Nash equilibrium. Similarly, the Price of Anarchy with respect to Bayes-Nash equilibria (BPoA) is the worst ratio between the optimal expected social welfare and the expected welfare in a mixed Bayes-Nash equilibrium. Formally, define $\textsf{NASH}(v)$ and $\textsf{BNASH}(D)$ as the set of all mixed Nash and mixed Bayes Nash equilibria respectively. Then,
\begin{align*}
	PoA &= \max_v \max_{b \in \textsf{NASH}(v)} \frac{\max\limits_{x \in \Omega} \sum_{i \in N} v_i(x)}{\Ex{\sum_{i \in N} v_i(f(b))}}, 
        \ \text{and}\\ 
	BPoA &= \max_D \max_{b \in \textsf{BNASH}(D)} \frac{\max\limits_{x \in \Omega} \Ex{\sum_{i \in N} v_i(x)}}{\Ex{\sum_{i \in N} v_i(f(b))}}.
\end{align*}

\paragraph{The Smoothness Framework.} An important ingredient in our result is the following notion of a smooth mechanism of Syrgkanis and Tardos \cite{SyrgkanisT13}. A mechanism is {\em $(\lambda,\mu)$-smooth} for $\lambda, \mu \ge 0$ if for all valuation profiles $v$ and all bid profiles $b$ there exists a possibly randomized strategy $b'_i$ for every agent $i$ that may depend on the valuation profile $v$ of all agents and the bid $b_i$ of that agent such that
\[
	\sum_{i \in N} \Ex{u_i((b'_i,b_{-i}),v_i)} \ge \lambda \cdot \max_{x \in \Omega} \sum_{i \in N} v_i(x) - \mu \cdot \sum_{i \in N} \Ex{p_i(b)}.
\]

\begin{theorem}[Syrgkanis and Tardos~\cite{SyrgkanisT13}]\label{thm:vasilis-and-eva}
If a mechanism is $(\lambda,\mu)$-smooth and agents have the possibility to withdraw from the mechanism, then the expected social welfare at any mixed Nash or mixed Bayes-Nash equilibrium is at least $\lambda/\max(\mu,1)$ of the optimal social welfare.
\end{theorem}

As shown in \cite{SyrgkanisT13}, $(\lambda,\mu)$-smoothness also implies a bound of $\max(\mu,1)/\lambda$ on the Price of Anarchy for correlated equilibria, also known as learning outcomes.  Furthermore, the simultaneous and sequential composition of multiple $(\lambda,\mu)$-smooth mechanisms is again $(\lambda,\mu)$-smooth. For details, on the precise definitions and statements beyond Nash equilibria, see \cite{SyrgkanisT13}.

In fact, our smoothness proofs show an even slightly stronger property, semi-smoothness as defined by \cite{Caragiannis++JET13}: the deviation strategy $b'_i$ only depends on the respective agent's valuation $v_i$, but not on the agent's bid $b_i$ or the other agents'  valuations $v_{-i}$. Therefore, the same Price of Anarchy bounds also apply to coarse correlated equilibria and Bayes-Nash equilibria with correlated types. 

\section{Oblivious Rounding and Smooth Relaxations}
\label{sec:meta-theorems}
In this section, we show our main theorem. We consider mechanisms for a problem $\Pi$ that are constructed as follows. First, one computes an optimal solution $x'$ to a relaxed problem $\Pi'$ that maximizes the declared welfare. That is, it maximizes $\sum_{i \in N} b_i(x')$. Afterwards, an $\alpha$-approximate oblivious rounding scheme is applied to derive a feasible solution $x$ to the original problem $\Pi$. Each bidder is charged $b_i(x)$, i.e., his declared value of this outcome.

\begin{theorem}[Main Result]\label{thm:main}
Consider problem $\Pi$ and a relaxation $\Pi'.$
Given a pay-your-bid mechanism $M' = (f', p')$ that is $(\lambda, \mu)$-smooth where $f'$ is an exact declared welfare maximizer for the relaxation $\Pi'$. Then a pay-your-bid mechanism $M=(f,p)$ for the original problem $\Pi$ that is obtained from the relaxation through an $\alpha$-approximate oblivious rounding scheme is $(\lambda/(2\alpha), \mu)$-smooth.
\end{theorem}

In many applications, smoothness is shown by the deviation strategy of reporting half one's true value. First we show that, while generally the deviation strategy $b_i'$ can be arbitrary, it is sufficient to consider only this deviation $b_i' = \frac{1}{2} v_i$. We exploit the fact that $f'$ performs exact optimization.

\begin{lemma}\label{lem:half}
Given a pay-your-bid mechanism $M = (f, p)$ that is $(\lambda, \mu)$-smooth where $f$ is an exact declared welfare maximizer. Then $M$ is $(\lambda/2,\mu)$-smooth for deviations to half the value. That is, for all bid vectors $b$ and bids $b'_i = \frac{1}{2}v_i$ for all $i \in N$,
$
\sum_{i \in N} u_i((b_i', b_{-i}), v_i) \geq \frac{\lambda}{2} OPT(v) - \mu \sum_{i \in N} p_i(b)
$.
\end{lemma}

\begin{proof}
We first use $(\lambda, \mu)$-smoothness of $M$. For any valuations, there have to be deviation bids fulfilling the respective conditions. So, in particular, let us pretend that each bidder $i$ has valuation $\frac{1}{2} v_i$. By smoothness, there are bids $b_i''$ against $b$ such that
\begin{equation}
\sum_{i \in N} u_i\left((b_i'', b_{-i}), \frac{1}{2} v_i\right)
\geq \lambda OPT\left( \frac{v}{2} \right) - \mu \sum_{i \in N} p_i(b). \label{eq:lemma1:smoothness}
\end{equation}

The next step is to relate the sum of utilities that agents with valuations $v$ get in $M$ when they unilaterally deviate from $b$ to $b'_i$, i.e., $\sum_{i \in N} u_i((b'_i,b_{-i}),v_i) = \sum_{i \in N} \frac{1}{2}v_i(f_i(b'_i,b_{-i})) = \sum_{i \in N} b'_i(f_i(b'_i,b_{-i}))$, to the sum of utilities that they get in $M$ with valuations $\frac{1}{2}v$ and unilateral deviations from $b$ to $b''_i$, i.e., $\sum_{i \in N} u_i((b''_i,b_{-i}),\frac{1}{2}v_i)$.

The allocation function $f$ optimizes exactly over its outcome space. Therefore, it can be used to implement a truthful mechanism $M^{\text{VCG}} = (f, p^{\text{VCG}})$ by applying VCG payments. As VCG payments are non-negative, we get
\[
u_i((b'_i,b_{-i}),v_i) = \frac{1}{2} v_i(f(b_i', b_{-i})) = b_i'(f(b_i', b_{-i})) \geq b_i'(f(b_i', b_{-i})) - p^{\text{VCG}}(b_i', b_{-i}).
\]
Observe that the latter term is exactly the utility bidder $i$ receives in $M^{\text{VCG}}$ if his valuation and bid is $b_i'$. As $M^{\text{VCG}}$ is truthful, this term is maximized by reporting the true valuation. In other words, it can only decrease, if bidder $i$ changes his bid to $b_i''$ (keeping the valuation $b_i'$). That is,
\[
u_i((b'_i,b_{-i}),v_i) \geq b_i'(f(b_i', b_{-i})) - p^{\text{VCG}}(b_i', b_{-i}) \geq b_i'(f(b_i'', b_{-i})) - p_i^{\text{VCG}}(b_i'', b_{-i}).
\]
Finally, we use that $p^{\text{VCG}}$ is no larger than $p$ because VCG payments never exceed bids, i.e., $p_i^{\text{VCG}}(b_i'', b_{-i}) \leq b_i''(f(b_i'', b_{-i})) = p_i(b_i'', b_{-i})$. By furthermore changing $b_i'$ back to $\frac{1}{2} v_i$, we get
\[
u_i((b'_i,b_{-i}),v_i) \geq \frac{1}{2} v_i(f(b_i'', b_{-i})) - p_i(b_i'', b_{-i})
= u_i\left((b''_i,b_{-i}), \frac{1}{2}v_{i}\right).
\]

Summing this inequality over all $i \in N$ and combining it with inequality
\eqref{eq:lemma1:smoothness}, we get
\[
\sum_{i \in N} u_i((b'_i,b_{-i}),v_i)
\geq \lambda OPT\left( \frac{v}{2} \right) - \mu \sum_{i \in N} p_i(b) = \frac{\lambda}{2} OPT(v)  - \mu \sum_{i \in N} p_i(b). \qedhere
\] 
\end{proof}

It remains to show that smoothness of the relaxation for deviations to half the value, implies smoothness of the derived mechanism for the original problem.
As it is often possible to directly show smoothness for deviations to half the value, we state the following stronger version of Theorem~\ref{thm:main} for relaxations that are $(\lambda,\mu)$-smooth for deviations to half the value.

Theorem \ref{thm:main} follows by first using Lemma~\ref{lem:half} to argue that unconstrained $(\lambda,\mu)$-smoothness of the relaxation implies $(\lambda/2,\mu)$-smoothness for deviations to half the value and then using Theorem~\ref{thm:smooth} to show that the derived mechanism is $(\lambda/(2\alpha),\mu)$-smooth.

\begin{thmbis}{thm:main}[Stronger Version of Main Theorem]\label{thm:smooth}
If the pay-your-bid mechanism $M' = (f',p')$ that solves the relaxation $\Pi'$ optimally is $(\lambda,\mu)$-smooth for deviations to $b'_i = \frac{1}{2} v_i$, then the pay-your-bid mechanism $M = (f,p)$ for $\Pi$ that is obtained from the relaxation through an $\alpha$-approximate oblivious rounding scheme is $(\lambda/\alpha,\mu)$-smooth.
\end{thmbis}
\begin{proof}
For any bid vector $b$, denote the utility of agent $i \in N$ under mechanism $M = (f,p)$ by $u_i(b, v) = v_i(f_i(b)) - p_i(b)$ and under mechanism $M' = (f',p')$ by $u'_i(b, v) = v_i(f'_i(b)) - p'_i(b).$

For each bidder $i$, we consider the unilateral deviation by $b_i' = \frac{1}{2} v_i$. As $M$ is a pay-your-bid mechanism, bidder $i$'s utility when bidding $b_i'$ against $b_{-i}$ can be expressed by 
\[
\Ex{u_i((b_i', b_{-i}), v_i)} = \Ex{v_i(f(b_i', b_i)) - p_i(b_i', b_{-i})} = \frac{1}{2} \Ex{v_i(f(b_i', b_{-i}))}.
\]
Next we use that the outcome $f(b_i', b_{-i})$ is derived from $f'(b_i', b_{-i})$ by applying an $\alpha$-approximate oblivious rounding scheme by considering the weight function in which $w_i = v_i$ for all $i$ and concluding that $\Ex{v_i(f(b_i', b_{-i}))} \geq \frac{1}{\alpha} v_i(f'(b_i', b_{-i}))$. That is, for bidder $i$'s utility, we get 
\[
\Ex{u_i((b_i', b_{-i}), v_i)} \geq \frac{1}{2 \alpha} v_i(f'(b_i', b_{-i})) = \frac{1}{\alpha} u_i'((b_i', b_{-i}), v_i),
\]
where the last step uses that $M'$ is a pay-your-bid mechanism as well.

Next, we apply the fact that $M'$ is $(\lambda,\mu)$-smooth for deviations to $b'_i = \frac{1}{2}v_i.$ We get for the sum of utilities in $M$ that
\[
\sum_{i \in N} \Ex{u_i((b_i', b_{-i}), v_i)} \geq \frac{1}{\alpha} \sum_{i \in N} u_i'((b_i', b_{-i}), v_i) \geq \frac{1}{\alpha} \left( \lambda OPT(v) - \mu \sum_{i \in N} p'_i(b) \right).
\]

To bound the terms $p'_i(b)$, we use once more the fact that we are applying an $\alpha$-approximate oblivious rounding scheme, this time to derive $f(b)$ from $f'(b)$ and considering the weight function in which $w_i = b_i$ for all $i$. This implies
\[
p'_i(b) = b_i(f'_i(b)) \leq \alpha \Ex{b_i(f_i(b))} = \alpha \Ex{p_i(b)}.
\]

Overall, we get
\[
\sum_{i \in N} \Ex{u_i((b_i', b_{-i}), v_i)} \geq \frac{1}{\alpha} \lambda OPT(v) - \mu \sum_{i \in N} p_i(b),
\]
as claimed.
\end{proof}

\section{Sparse Packing Integer Programs}
\label{section:multiunit}
In a sparse packing integer program (PIP) each bidder $i$ can be served in $K$ possible ways. The fact whether bidder $i$ gets option $k$ is represented by a binary variable $x_{i, k} \in \{0,1\}$. Each bidder $i$ can only get one option, that is $\sum_{k \in [K]} x_{i, k} \leq 1$ for each $i$. Furthermore, matrix $A$ and vector $c$ represent packing constraints between the bidders, requiring that $A x \leq c$. Each bidder's valuation depends on the option that he is served by. That is $v_i$ can be expressed as $v_i(x) = \sum_{k \in [K]} v_{i, k} x_{i, k}$. The goal is to find $\max \sum_{i \in N} v_i(x)$ subject to feasibility.

We consider the relaxation of this integer program in which the binary variables $x_{i, k} \in \{0,1\}$ are replaced with non-negative variables $x_{i, k} \ge 0$. The interpretation is that $x_{i, k}$ is a fractional allocation of option $k$ to bidder $i$, and no bidder $i$ can be assigned more than the fractional equivalent of one option. This relaxation is a LP and can therefore be solved in polynomial time.

The \emph{column sparsity} $d$ is the maximum number of non-zero entries in a single column of $A$. Formally, for each variable $x_{j, k}$, let $S_{j, k}$ be the set of constraints in $A$ with a non-zero coefficient, that is, $S_{j, k} = \{ \ell \mid A_{\ell, j, k} \neq 0 \}$. Now $d = \max_{j, k} \lvert S_{j, k} \rvert$. Examples with $d = 1$ are multi unit-auctions with unconstrained valuations or unit demand auctions, where each player wants at most one item, possibly with player dependent capacity constraints, like makespan constraints in a generalized assignment problem; or more generally, combinatorial auctions in which each bidder is interested in bundles of at most $d$ items are an example with $d \ge 1$.

\begin{theorem}\label{thm:pip-poa}
There is an oblivious rounding based, pay-your-bid mechanism for $d$-sparse packing integer programs that achieves a Price of Anarchy of $32$ for $d =1$ and of $16d(d+1)$ for general $d$.
\end{theorem}

To prove Theorem~\ref{thm:pip-poa} we use the fact that an $8d$-approximate oblivious rounding scheme is available through \cite{BansalKNS10}. In addition, we show in Lemma~\ref{lem:pip-smooth} that the canonical LP relaxation of a $d$-sparse packing integer program that this rounding scheme is based on is $(1/2,d+1)$-smooth for deviations to $b'_i = \frac{1}{2}v_i$. The claimed bound on the Price of Anarchy then follows from Theorem~\ref{thm:smooth}.

\begin{lemma}\label{lem:pip-smooth}
The pay-your-bid mechanism that solves the canonical LP relaxation of a $d$-sparse packing integer program is $(1/2,d+1)$-smooth for deviations to $b'_i = \frac{1}{2}v_i.$
\end{lemma}

To prove this lemma we show the following auxiliary lemma. Given a bid vector $b$ and a capacity vector $c$, we denote by $W^b(c)$ the value of the optimal LP solution.

\begin{lemma}\label{lem:lemma-4.7-lps}
Let $\bar{x}$ be an arbitrary fractional solution. Then,
\begin{equation*}
	\sum_{i \in N} \left(W^{b_{-i}}(c) - W^{b_{-i}}(c - A (\bar{x}_i, 0) ) \right)\leq (d + 1) \cdot W^{b}(c).
\end{equation*}
\end{lemma}
\begin{proof}
Let $\hat{x}$ denote a fractional allocation that maximizes declared welfare for players $N$ with bids $b$ (we have $0 \leq \hat{x}_{i, k}$, $\sum_{k} \hat{x}_{i, k} \leq 1$, $A \hat{x} \leq c$).

Now, define LP solution $\hat{x}^{-i}$ by setting $\hat{x}^{-i}_{j, k} = \left( 1 - \delta^i_{j, k} \right) \hat{x}_{j, k}$, where $\delta^i_{j, k} = \max_{\ell \in S_{j, k}} \frac{(A (\bar{x}_i, 0))_\ell}{c_\ell}$. Observe that $A \hat{x}^{-i} \leq c - A (\bar{x}_i, 0)$.

For all $j$ and $k$, we have $\sum_i \delta^i_{j, k} \leq \sum_i \sum_{\ell \in S_{j, k}} \frac{(A (\bar{x}_i, 0))_\ell}{c_\ell} = \sum_{\ell \in S_{j, k}} \sum_i \frac{(A (\bar{x}_i, 0))_\ell}{c_\ell} \leq \lvert S_{j, k} \rvert \leq d$ and therefore $\sum_{i \neq j, i \in N} (1 - \delta^i_{j, k}) \geq n - d - 1$. This gives us
\begin{align*}
	\sum_{i \in N} W^{b_{-i}}(c - A (\bar{x}_i, 0))
	&\ge \sum_{i \in N} \sum_{j \neq i, j \in N} \sum_k b_{j, k} \hat{x}^{-i}_{j, k} \\
	&= \sum_{j \in N} \sum_{i \neq j, i \in N} \sum_k b_{j, k} \left( 1 - \delta^i_{j, k}\right) \hat{x}_{j, k} \\
	&\geq \left( n - d - 1 \right) \sum_{j \in N}  \sum_k b_{j, k} \hat{x}_{j, k} \\
	&= (n - d - 1) W^b(c), 
\end{align*}
which gives the claimed bound as clearly $W^{b_{-i}}(c) \le W^{b}(c)$.
\end{proof}

\begin{proof}[Proof of Lemma~\ref{lem:pip-smooth}]
Consider valuations $v$, bids $b$ and deviations of each player $i \in N$ to $b'_i = 1/2 \cdot v_i.$
Denote the optimal fractional allocation for bids $(b'_i,b_{-i})$ by $\bar{x}_1(b'_i,b_{-i}),\dots,\bar{x}_n(b'_i,b_{-i}).$
Then, by the definition of $b'_i$,
\begin{align*}
	u_i((b'_i,b_{-i}),v_i) = v_i(\bar{x}_i(b'_i,b_{-i})) - b'_i(\bar{x}_i(b'_i,b_{-i})) = b'_i(\bar{x}_i(b'_i,b_{-i})).
\end{align*}
Since $\bar{x}_1(b'_i,b_{-i}),\dots,\bar{x}_n(b'_i,b_{-i})$ is fractional allocation that maximizes declared welfare with respect to bids $(b'_i,b_{-i})$,
\begin{align*}
	b'_i(\bar{x}_i(b'_i,b_{-i})) + W^{b_{-i}}(c)
	&\ge b'_i(\bar{x}_i(b'_i,b_{-i})) + \sum_{j \neq i} b_j(\bar{x}_j(b'_i,b_{-i})) \displaybreak[0]\\
	&\ge b'_i(\bar{x}_i(v)) + W^{b_{-i}}(c - A (\bar{x}_i(v), 0)).
\end{align*}
Rearranging this gives
\begin{align*}
	b'_i(\bar{x}_i(b'_i,b_{-i}))
	&\ge b'_i(\bar{x}_i(v)) - [W^{b_{-i}}(c) - W^{b_{-i}}(c - A (\bar{x}_i(v), 0))].
\end{align*}
Summing over all players and applying Lemma \ref{lem:lemma-4.7-lps}, we obtain
\begin{align*}
	\sum_{i \in N} u_i((b'_i,b_{-i}),v_i) =
	&\sum_{i \in N} b'_i(\bar{x}_i(b'_i,b_{-i})) \displaybreak[0] \\
	&\ge \sum_{i \in N} (b'_i(\bar{x}_i(v)) - [W^{b_{-i}}(c) - W^{b_{-i}}(c - A (\bar{x}_i(v), 0))].) \displaybreak[0] \\
	&\ge \sum_{i \in N} b'_i(\bar{x}_i(v)) - (d+1) \cdot \sum_{i \in N} b_i(\bar{x}_i(b)) \displaybreak[0] \\
	&= \frac{1}{2} \cdot \sum_{i \in N} v_i(\bar{x}_i(v)) - (d+1) \cdot \sum_{i \in N} b_i(\bar{x}_i(b)),
\end{align*}
which completes the proof.
\end{proof}

We conclude this section with the observation that it is crucial to take the detour via relaxation and rounding: 
The mechanism that solves the integral problem optimally has an unbounded Price of Anarchy even when $d = 1$.

\begin{proposition}
\label{proposition:multiunit:integralunsmooth}
The pay-your-bid mechanism that maximizes the declared welfare over integral allocations in a multi-unit auction has a pure Nash equilibrium whose welfare is by a factor $(n-2)/2 = m/2$ smaller than the optimal welfare, where $n$ is the number of bidders and $m$ is the number of goods. 
\end{proposition}
\begin{proof}
Consider a setting with $n$ bidders and $m = n - 2$ units of an identical good.  The valuations for bidders $i = 1, \ldots, n-2$ are $v_{i, 0} = 0$ and $v_{i, k} = 1$ for all $k \geq 1$. For bidders $i = n - 1$ and $i = n$, we set $v_{i, k} = 0$ for $k < m$ and $v_{i, m} = 2$. It is socially optimal to allocate one item each to bidders $1, \ldots, m$, achieving welfare $m$.

We observe that in this setting, the following bids are a pure Nash equilibrium. For $i = 1, \ldots, m$ we have $b_{i, k} = 0$ for all $k$ and for $i = m+1, m+2$ we have $b_{i, k} = v_{i, k}$ for all $k$. This equilibrium has social welfare $2$.\qquad
\end{proof}

\section{Single Source Unsplittable Flow}
\label{sec:flow}

We consider the single source weighted unsplittable multi-commodity flow problem in which we are given a graph $G = (V,E)$ with edge capacities $c_e$ for each edge $e \in E$. All bidders share a source node $s$ and each bidder $i$ has a sink node $t_i$. He asks for a path connecting $s$ and $t_i$ fulfilling his demand $d_i$. His value for this is $v_i$, and he has no value for less flow than his demand. We assume that the sink $t_i$ and demand $d_i$ for each player is common knowledge, so the player's bid is a claimed value, which will be denoted by $b_i$.

Let $\mathcal{P}_i$ be the paths connecting $s$ and $t_i$. For each $P \in \mathcal{P}_i$, we have a variable $f_{i, P}$ denoting the amount of flow along path $P$. The problem requires single path routing, that is, all the $d_i$ flow satisfying player $i$'s demand must be carried by a single path. We use the following standard LP relaxation that maximizes $\sum_{i \in N} \sum_{P \in \mathcal{P}_i} b_i f_{i, P}$ subject to $\sum_{i \in N} \sum_{P \in \mathcal{P}_i: e \in P} f_{i, P} \leq c_e$ for all $e \in E$ and $\sum_{P \in \mathcal{P}_i} f_{i, P} \leq d_i$ for all $i \in N$.

Substituting $f_{i, P}$ by $d_i \bar{x}_{i, P}$, we get an LP formulation in the spirit of Section~\ref{section:multiunit}. However, this LP is not necessarily sparse, as the column sparsity $d$ corresponds to the maximum path length.
Nevertheless we are able to establish the following theorem.

\begin{theorem}\label{thm:flow}
Suppose the minimum edge capacity is by a logarithmic factor larger than the maximum demand, i.e., $\min_{e \in E} c_e \ge c \epsilon^{-1} \log |E| \max_{i \in N} d_i$ for some $\epsilon > 0$ and an appropriate constant $c>0$. Then there is an oblivious rounding based, pay-your-bid mechanism for the single source unsplittable flow problem with Price of Anarchy at most $2(1+\epsilon)$.
\end{theorem}

For the setting considered here Raghavan and Thompson \cite{Raghavan1988,RaghavanThompson1987} present a $(1+\epsilon)$-approximate oblivious rounding scheme. We will argue below that the canonical LP relaxation that this rounding scheme is based on can be solved exactly using a greedy heuristic, which is $(1/2,1)$-smooth for deviations to $b'_i = \frac{1}{2}v_i.$ The Price of Anarchy bound then follows by Theorem~\ref{thm:smooth}.

To show smoothness we proceed in two steps. We first consider the integral problem and the special case when all demands are $1$. Afterwards we consider the fractional problem with arbitrary demands.

It is not hard to see that in the integral problem when all demands are $1$, the subsets of terminals whose demand can be routed defines a matroid, known as gammoid (see \cite{Frank}). We start by showing that the greedy algorithm for matroids is $(\frac{1}{2},1)$-smooth for deviations to $b'_i = v_i/2$.

\begin{lemma}
The greedy algorithm for matroids is  $(\frac{1}{2},1)$-smooth for deviations to half the value.
\end{lemma}
\begin{proof}
Consider an optimal independent set $I$, and a bid profile $b$. Let $J$ be the basis selected by the greedy algorithm on bids $b$. For two bases of a matroid $I$ and $J$, there is a matching between the elements so that for all $i \in I$ the set $I-i+m(i)$ is a basis, where $m(i)\in J$ is the element matched to $i$. Using this matching, we get that for any $i\in I$ the deviation $b'_i=v_i/2$ satisfies $u_i(b_i',b_{-i}) \ge \frac{1}{2} v_i -b_{m(i)}$. To see why this is true, we consider two cases. If $b_i'=v_i/2>b_{m_i}$ then the greedy algorithm on bids $(b_i',b_{-i})$ will consider $i$ before $m(i)$ and will add $i$ to the independent set, so $u_i(b_i',b_{-i}) = v_i-b_i=\frac{1}{2}v_i$, and hence the inequality follows. If $b_i'=v_i/2 \le b_{m_i}$ then the right hand side is non-positive, and we have $u_i(b_i',b_{-i})\ge 0$ so the inequality is true in this case also. For $i \notin I$ clearly $u_i(b_i',b_{-i}) \ge 0$. summing over all $i$ we get the claimed smoothness bound.

For completeness, we also include the proof that the matching used above exists. For a subset $A\subset I$ let $e(A)$ be the set of elements in $J$ that can be exchanged with an element in $A$. By Hall's theorem we only need to prove that $|A| \le |e(A)|$ for all $A$. Suppose this is not true, and let $A$ be a set with $|A| >|e(A)|$. By the matroid exchange property, $I\setminus A$ can be extended to a basis from elements of $J$. Let $B \subset J$ denote the set that is added to $I\setminus A$, so $J'=(I \setminus A) \cup B$ is a basis. We claim that $B \subset e(A)$. To see this for an element $j \in B$ use the matroid exchange property to extend the set $(I \setminus A)+j$ to a basis of the matroid from the set $I$. The element left out must be an element in $A$ and has now been exchanged for $j$.
\end{proof}

Now consider the fractional flow problem. In case of unit demands, we have just seen that the greedy algorithm is $(1/2,1)$ smooth, and it solved the integer problem. Notice that the greedy algorithm is simply the Ford-Fulkerson augmenting paths algorithm prioritizing terminals with higher bids $b_i$. Now consider the problem with arbitrary demands. We claim that the greedy algorithm using repeated augmenting paths to find the maximum flow, prioritizing terminals with higher $b_i/d_i$ (willingness to pay per unit of flow) optimally solves the fractional problem, and is $(\frac{1}{2},1)$-smooth for deviations to $b'_i = v_i/2$. For the case when demands are integers this can be thought of as a corollary of the above lemma for matroids, and can be proved analogously in the general case.

\begin{lemma}
The above fractional routing algorithm optimally solves the fractional single source routing problem, and is $(\frac{1}{2},1)$-smooth for deviations to half the value.
\end{lemma}

Importantly, the reference to greedy in the above proof is on the fractional level, as the greedy algorithm for the integral problem with general demands can be as bad as an $\Omega(\sqrt{|E|})$ (see \cite{Kleinberg96}). Also, as we show next, solving the integral problem optimally again leads to an unbounded Price of Anarchy, even if there is a single source, a single target and just one unit capacity edge between the two.

\begin{proposition}
The pay-your-bid mechanism that solves the integral single source unsplittable flow problem optimally has a Price of Anarchy of at least $\frac{m}{2}$ even if there is a single source, a single target and a unit-capacity edge connecting the two.
\end{proposition}
\begin{proof}
 Consider a network with two nodes, one of which is the source and the other is the target, and a single edge with unit capacity between the two.
There are $m+2$ players. Two big players with demand $1$ and value $2$, and $m$ small players with demand $1/m$ and value $1$.

Having the big players both bid $2$ and the small players bid $0$ is a pure Nash equilibrium with social welfare $2$, where the optimal is $m$.\qquad
\end{proof}

\section{Max-TSP}
\label{sec:max-tsp}

In the asymmetric maximization version of the traveling salesperson problem, one is given a complete digraph $G = (V, E)$ with non-negative weights $(w_e)_{e \in E}$. Players are the edges with value $w_e$ for being selected, and the mechanism aims to select a Hamiltonian cycle $C$ that maximizes $\sum_{e \in C} w_e$. We show how existing combinatorial algorithms for this problem can be interpreted as relax-and-round algorithms, and derive the following theorem.

\begin{theorem}
There is a pay-your-bid mechanism for the maximum traveling salesman problem based on oblivious rounding that achieves a Price of Anarchy of $9$.
\end{theorem}

We present a proof based on the algorithm of Fisher \cite{Fisher1979} that yields a Price of Anarchy bound of $12$; in Appendix \ref{app:tsp} we show how to improve this bound using the algorithm of Paluch et al.~\cite{PaluchEZ12} instead.

We will first argue how the algorithm of Fisher \cite{Fisher1979} can be interpreted as an oblivious, $2$-approximate rounding scheme that relaxes the problem of finding a maximum-weight Hamiltonian cycle to the problem of finding a maximum-weight cycle cover (defined below). We will then argue that the pay-your-bid mechanism that finds a cycle cover is $(1/2,3)$-smooth for deviations to half the value. Together with Theorem \ref{thm:smooth} these two facts imply the claimed Price of Anarchy bound.  

Fisher's algorithm uses cycle covers as relaxed solutions. A collection of cycles $C_1, \ldots, C_k$ in a (di-)graph is called a \emph{cycle cover} if each vertex of the graph is contained in exactly one of the cycles. A maximum-weight cycle cover can be computed in polynomial time by computing a maximum-weight perfect matching in a suitably defined bipartite graph. In order to approximate the max-weight TSP tour, one first determines a max-weight cycle cover $C_1, \ldots, C_k$, and then from each of the obtained cycles the minimum-weight edge is dropped, resulting in a collection of vertex-disjoint paths $P_1, \ldots, P_k$. These paths are connected in an arbitrary way to obtain a Hamiltonian cycle $C$. Going from $C_i$ to $P_i$, we lose at most half of the weight of this respective cycle. As all weights are non-negative, no weight is lost going from $P_1, \ldots, P_k$ to $C$. In combination, we have $\sum_{e \in C} w_e \geq \sum_{i \in [k]} \sum_{e \in P_i} w_e \geq \frac{1}{2} \sum_{i \in [k]} \sum_{e \in C_i} w_e$.

This ``rounding'' algorithm can also be modified to work in an oblivious way without loss in the worst case by removing one edge uniformly at random from each cycle. This way, for each edge that was contained in the cycle cover, the individual probability to be also included in the output is at least $\frac{1}{2}$.

To be able to apply Theorem \ref{thm:smooth} and obtain the Price of Anarchy bound it remains to show that the pay-your-bid mechanism that finds a cycle cover is $(1/2,3)$-smooth for deviations to half the value. 
\begin{lemma}\label{lem:cc}
The pay-your-bid mechanism for computing an optimal cycle cover is $(1/2,3)$-smooth for deviations to $b'_i = \frac{1}{2}v_i.$
\end{lemma}

Our proof of this lemma follows a similar pattern as our proof for sparse packing integer programs. The idea is again to bound the net loss in declared welfare for a given feasible allocation relative to the optimal declared welfare. Denote by $\mathcal{C}$ the set of all cycle covers and by $\mathcal{C}_{e}$ the set of all cycle covers that include edge $e \in E.$ Then $W^{b_{-i}}(\mathcal{C}) - W^{b_{-i}}(\mathcal{C}_{e})$ is the social cost of including $e \in E$ in the cycle cover. Specifically, we show the following auxiliary lemma.

\begin{lemma}\label{lem:cc-smooth}
Consider bids $b$. Let $C_1, \dots, C_k$ be the cycle cover that maximizes reported welfare for bids $b$ and use $E_{C}$ to denote the set of edges used in this cycle cover. Consider any other cycle cover $C'_1,\dots,C'_\ell$ with edge set $E_{C'}$. Then,
\[
	\sum_{i \in N: e_i \in E_{C'}} \left(W^{b_{-i}}(\mathcal{C}) - W^{b_{-i}}(\mathcal{C}_{e_i}) \right) \le 3 \cdot W^b(\mathcal{C}).
\]
\end{lemma}
\begin{proof}
Let us first consider any fixed edge $e \in E$ and let us construct a cycle cover from the edge set $E_C$ that contains $e$.

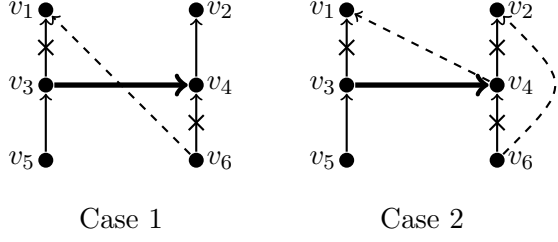
\begin{figure}
\center
\scalebox{1.0}{
\begin{tikzpicture}[thick]%
    \draw (1,-0.5) node[below](case2-label){Case 1};
    \draw (0,0) node[circle,fill,inner sep=2pt](v5){};
    \draw (0,0) node[left](v5-label){$v_5$};
    \draw (0,1) node[circle,fill,inner sep=2pt](v3){};
    \draw (0,1) node[left](v3-label){$v_3$};
    \draw (0,2) node[circle,fill,inner sep=2pt](v1){};
    \draw (0,2) node[left](v1-label){$v_1$};
    \draw (2,0) node[circle,fill,inner sep=2pt](v6){};
    \draw (2,0) node[right](v6-label){$v_6$};
    \draw (2,1) node[circle,fill,inner sep=2pt](v4){};
    \draw (2,1) node[right](v4-label){$v_4$};
    \draw (2,2) node[circle,fill,inner sep=2pt](v2){};
    \draw (2,2) node[right](v2-label){$v_2$};
    \draw[->, line width=2pt] (v3) -- (v4);
    \draw[->] (v5) -- (v3);
    \draw[->] (v3) -- (v1);
    \draw[->] (v6) -- (v4);
    \draw[->] (v4) -- (v2);
    \draw[dashed,->] (v6) -- (v1);
    \draw (-0.1,1.6) -- (0.1,1.4);
    \draw (-0.1,1.4) -- (0.1,1.6);
    \draw (1.9,0.6) -- (2.1,0.4);
    \draw (1.9,0.4) -- (2.1,0.6);
    \draw (5,-0.5) node[below](case2-label){Case 2};
    \draw (4,0) node[circle,fill,inner sep=2pt](v5){};
    \draw (4,0) node[left](v5-label){$v_5$};
    \draw (4,1) node[circle,fill,inner sep=2pt](v3){};
    \draw (4,1) node[left](v3-label){$v_3$};
    \draw (4,2) node[circle,fill,inner sep=2pt](v1){};
    \draw (4,2) node[left](v1-label){$v_1$};
    \draw (6,0) node[circle,fill,inner sep=2pt](v6){};
    \draw (6,0) node[right](v6-label){$v_6$};
    \draw (6,1) node[circle,fill,inner sep=2pt](v4){};
    \draw (6,1) node[right](v4-label){$v_4$};
    \draw (6,2) node[circle,fill,inner sep=2pt](v2){};
    \draw (6,2) node[right](v2-label){$v_2$};
    \draw[->, line width=2pt] (v3) -- (v4);
    \draw[->] (v5) -- (v3);
    \draw[->] (v3) -- (v1);
    \draw[->] (v6) -- (v4);
    \draw[->] (v4) -- (v2);
    \draw[dashed,->] (v4) -- (v1);
    \draw[dashed,->] (v6) .. controls(7,1) .. (v2);
    \draw (3.9,1.6) -- (4.1,1.4);
    \draw (3.9,1.4) -- (4.1,1.6);
    \draw (5.9,0.6) -- (6.1,0.4);
    \draw (5.9,0.4) -- (6.1,0.6);
    \draw (5.9,1.6) -- (6.1,1.4);
    \draw (5.9,1.4) -- (6.1,1.6);
\end{tikzpicture}
}
\vspace{-5pt}
\caption{Smoothness of Cycle Cover LP}\label{fig:cycle-cover}
\vspace{-5pt}
\end{figure}

Figure \ref{fig:cycle-cover} depicts the two possible cases. The thin edges are edges from $E_C$. The thick edge is $e$. Note that we can w.l.o.g.~assume that $v_1$ and $v_4$ are distinct. (As otherwise $e$ would already be contained in $E_C$ and there would be nothing to show.)

The first case is when nodes $v_1$ and $v_6$ are distinct. In this case we can remove edges $(v_3,v_1)$ and $(v_6,v_4)$ from $E_C$ and add edge $(v_6,v_1)$. The resulting edge set is a valid cycle cover because this modification to $E_C$ maintains the in-/outdegrees of all nodes and does not create self-loops as we assumed $v_1$ and $v_6$ to be distinct.

The second case is when nodes $v_1$ and $v_6$ are identical. In this case the modification just described would create a self-loop from/to node $v_1 = v_6$. We can avoid this by instead removing edges $(v_3,v_1)$, $(v_6,v_4)$, and $(v_4,v_2)$ from $E_C$ and adding edges $(v_4,v_1)$ and $(v_6,v_2)$. This leads to a valid cycle cover as it maintains in- and outdegrees and does not create self-loops as we assumed $v_1$ and $v_6$ to be identical and $v_2$ must be different from $v_1$ (as otherwise this node would have an in degree of $2$, which would violate the cycle cover constraint).

In order to bound $\sum_{i \in N: e_i \in E_{C'}} \left(W^{b_{-i}}(\mathcal{C}) - W^{b_{-i}}(\mathcal{C}_{e_i}) \right)$, we now assume that $e$ was drawn uniformly at random from $E_{C'}$, the set of all edges in $C'$. As edge weights are non-negative, the loss in declared welfare by forcing $e$ into $C$, i.e., $W^{b_{-i}}(\mathcal{C}) - W^{b_{-i}}(\mathcal{C}_{e})$ is upper-bounded by the weight of all edges removed by our construction. For any edge $e'$, the probability of being removed by our construction is at most $\frac{3}{\lvert E_{C'} \rvert}$. This is due to the fact that $e'$ is only removed if it fulfills a certain role in relation to $e$, namely being the edge $(v_3, v_1)$, $(v_4, v_2)$, or $(v_6, v_4)$ in Figure~\ref{fig:cycle-cover}. As $C'$ is a cycle cover, each edge $e'$ can have each role only with respect to a single $e \in E_{C'}$. Overall, this gives
\[
\Ex{W^{b_{-i}}(\mathcal{C}) - W^{b_{-i}}(\mathcal{C}_{e})} \leq \sum_{e' \in E_{C}} w_{e'} \frac{3}{\lvert E_{C'} \rvert} = \frac{3}{\lvert E_{C'} \rvert} W^b(\mathcal{C}).
\]
Using the definition of the expectation, and multiplying the inequality by $\lvert E_{C'} \rvert$ shows the claim.
\end{proof}

\begin{proof}[Proof of Lemma~\ref{lem:cc}]
Following the same steps as in the proof of Lemma~\ref{lem:pip-smooth} and using Lemma~\ref{lem:cc-smooth} instead of Lemma~\ref{lem:lemma-4.7-lps} completes the proof.
\end{proof}

\section{Combinatorial Auctions}
\label{sec:cas}

In this section, we consider combinatorial auctions (CAs). In a CA, $m$ items are sold to $n$ bidders. Each item is allocated to at most one bidder and each bidder $i$ has a valuation $v_i(S)$ for the subset $S \subseteq [m]$ of items he receives. The canonical relaxation as a {\em configuration LP}  uses variables $x_{i, S} \in [0, 1]$ representing the fraction that bidder $i$ receives of set $S$. The goal is to maximize $\sum_{i \in N} \sum_{S \subseteq [m]} b_i(S) x_{i, S}$ s.t.~$\sum_{i \in N} \sum_{S: j \in S} x_{i, S} \leq 1$ for all $j \in [m]$ and $\sum_{S} x_{i, S} \leq 1$ for all $i \in N$.

For arbitrary valuation functions, only very poor approximation factors can be achieved for the optimization problem. Therefore, we focus on XOS or fractionally subadditive valuations. That is, each valuation function $v_i$ has a representation of the following form. There are values $v_{i, j}^\ell \geq 0$ such that $v_i(S) = \max_\ell \sum_{j \in S} v_{i, j}^\ell$.
Feige et al.~\cite{FeigeFIILS14} very recently generalized the class of XOS functions to $\mathcal{MPH}$-$k$, where XOS is precisely the case $k = 1$. A valuation function $v_i$ belongs to class $\mathcal{MPH}$-$k$ if there are values $v_{i, T}^\ell \geq 0$ such that $v_i(S) = \max_\ell \sum_{T \subseteq S, \lvert T \rvert \leq k} v_{i, T}^\ell$.

\begin{theorem}\label{thm:poa-ca}
There is a pay-your-bid mechanism that is based on oblivious rounding and achieves a Price of Anarchy of $4\frac{e}{e-1}$ for XOS-valuations and of $O(k^2)$ for $\mathcal{MPH}$-$k$-valuations.
\end{theorem}

For general $\mathcal{MPH}$-$k$-valuations \citet{FeigeFIILS14} present an $O(k+1)$-approximate oblivious rounding scheme; a better constant of $\frac{e}{e-1}$ for the special case XOS  can be achieved via the rounding scheme described in \citet{Feige09}. 
Regarding smoothness we show below that the configuration LP that these rounding schemes are based on is $(1/2,k+1)$-smooth for deviations to $b'_i = \frac{1}{2} v_i$. The claimed Price of Anarchy bounds then follow from Theorem~\ref{thm:smooth}.

As in the previous applications the key lemma in the smoothness proof is the following lemma that bounds the net loss of enforcing a feasible solution one player at a time.
For a bid profile $b$ and a vector of quantities $q$ let $W^b(q)$ denote the optimal declared social welfare over all fractional allocations, constrained by capacity vector $q$.


\begin{lemma}\label{lem:lemma4.7-mphk}
Let $x$ be an arbitrary fractional solution to the configuration LP. Then, 
\[ 
	\sum_{i \in N} \left(W^{b_{-i}}(1) - W^{b_{-i}}(1 - x_i ) \right)\leq (k + 1) \cdot W^{b}(1).
\]
\end{lemma}
\begin{proof}
As removing a player can only decrease the declared welfare that can be achieved, we have $\sum_{i \in N} W^{b_{-i}}(1) \le \sum_{i \in N} W^b(1)$. 
It now remains to show that $\sum_{i \in N} W^{b_{-i}}(1-x_i) \ge (n-k-1) \sum_{i \in N} W^b(1)$. Subtracting this inequality from the first one implies then the claim.

Let $\hat{x}$ denote a fractional allocation that maximizes declared welfare for players $N$ with bids $b$, so $\sum_S b_i(S) \hat{x}_{i, S} = W^b(1)$. Let $\hat{b}_{i, T}$ be the corresponding values such that 
\[
W^b(1) = \sum_S b_i(S) \hat{x}_{i, S} = \sum_S \left( \sum_{T \subseteq S, \lvert T \rvert \leq k} \hat{b}_{i, T} \right) \hat{x}_{i, S} \enspace.
\]

In order to bound $W^{b_{-i}}(1-x_i)$ for a fixed $i$, we will turn $\hat{x}$ into a feasible solution $\hat{x}^{-i}$ for the more restricted constraint capacities $1-x_i$. To simplify notation, we assume that for every $i$ we have $\sum_{A} x_{i, A} = 1$. This is possible without loss of generality as we can increase $x_{i, \emptyset}$ without modifying the objective function or feasibility. The LP solution $\hat{x}^{-i}$ is now defined by setting $\hat{x}^{-i}_{i', S} = \sum_{A} x_{i, A} \sum_{U: S = U \setminus A} \hat{x}_{i', U}$. 

The first step is to show feasibility of this solution. For all $A \subseteq [m]$ and $j \in [m]$, we have
\[
\sum_{i' \neq i} \sum_{S: j \in S} \sum_{U: S = U \setminus A} \hat{x}_{i', U} = \begin{cases} 0 & \text{ if $j \in A$} \\ \sum_{i' \neq i} \sum_{U: j \in U} \hat{x}_{i', U} & \text{ if $j \not\in A$}.
\end{cases}
\]
This implies that for all $j \in [m]$
\begin{align*}
\sum_{i' \neq i} \sum_{S: j \in S} \hat{x}^{-i}_{i', S} & = \sum_{i' \neq i} \sum_{S: j \in S} \sum_{A} x_{i, A} \sum_{U: S = U \setminus A} \hat{x}_{i', U} \\
& = \sum_{A} x_{i, A} \sum_{i' \neq i} \sum_{S: j \in S} \sum_{U: S = U \setminus A} \hat{x}_{i', U} \\
& = \sum_{A: j \not\in A} x_{i, A} \sum_{i' \neq i} \sum_{U: j \in U} \hat{x}_{i', U} \enspace.
\end{align*}
By feasibility of $\hat{x}$, we have $\sum_{i' \neq i} \sum_{S: j \in S} \sum_{U: S = U \setminus A} \hat{x}_{i', U}$. Furthermore, as we assumed $\sum_A x_{i, A} = 1$, we also have $\sum_{A: j \not\in A} x_{i, A} = 1 - \sum_{A: j \in A} x_{i, A}$. Therefore $\sum_{i' \neq i} \sum_{S: j \in S} \hat{x}^{-i}_{i', S} \leq 1 - \sum_{A: j \in A} x_{i, A}$. That is, $\hat{x}^{-i}$ is a feasible solution with respect to the capacity vector $q = 1 - x_i$.

Next, we bound the value of this constructed solution $\hat{x}^{-i}$. Let us first consider the contribution to the declared welfare by player $i' \neq i$ in $\hat{x}^{-i}$. We get
\begin{align*}
\sum_S b_{i'}(S) \hat{x}^{-i}_{i', S} & \geq \sum_S \left( \sum_{T \subseteq S, \lvert T \rvert \leq k} \hat{b}_{i', T} \right) \hat{x}^{-i}_{i', S}\\
& = \sum_S \left( \sum_{T \subseteq S, \lvert T \rvert \leq k} \hat{b}_{i', T} \right) \sum_{A} x_{i, A} \sum_{U: S = U \setminus A} \hat{x}_{i', U} \\
& = \sum_{A} x_{i, A} \sum_U \left( \sum_{T \subseteq U \setminus A, \lvert T \rvert \leq k} \hat{b}_{i', T} \right) \hat{x}_{i', U} \enspace.
\end{align*}
Taking the sum over all $i' \neq i$, this implies
\[
W^{b_{-i}}(1-x_i) \geq \sum_{i' \neq i} \sum_{A} x_{i, A} \sum_U \left( \sum_{T \subseteq U \setminus A, \lvert T \rvert \leq k} \hat{b}_{i', T} \right) \hat{x}_{i', U} \enspace.
\]

In the remainder, we will bound the sum of all $W^{b_{-i}}(1-x_i)$ and bound it in terms of $W^b(1)$. Using the bound on $W^{b_{-i}}(1-x_i)$ obtained so far and reordering the sums, we get
\begin{align*}
\sum_{i} W^{b_{-i}}(1-x_i) & \geq \sum_{i'} \sum_{i \neq i'} \sum_{A} x_{i, A} \sum_U \left( \sum_{T \subseteq U \setminus A, \lvert T \rvert \leq k} \hat{b}_{i', T} \right) \hat{x}_{i', U} \\
& = \sum_{i'} \sum_U \left( \sum_{i \neq i'} \sum_{A} x_{i, A} \sum_{T \subseteq U \setminus A, \lvert T \rvert \leq k} \hat{b}_{i', T} \right) \hat{x}_{i', U} \enspace.
\end{align*}
By reordering the sums further, we get
\begin{align*}
\sum_{i \neq i'} \sum_{A} x_{i, A} \sum_{T \subseteq U \setminus A, \lvert T \rvert \leq k} \hat{b}_{i', T} & = \sum_{T \subseteq U, \lvert T \rvert \leq k} \hat{b}_{i', T}\sum_{i \neq i'} \sum_{A: A \cap T = \emptyset} x_{i, A} \\
& = \sum_{T \subseteq U, \lvert T \rvert \leq k} \hat{b}_{i', T} \left( \sum_{i \neq i'} \sum_{A} x_{i, A} - \sum_{i \neq i'} \sum_{A: A \cap T \neq \emptyset} x_{i, A} \right) \enspace.
\end{align*}
As we assumed $\sum_{A} x_{i, A} = 1$ for all $i$, we have
\[
\sum_{i \neq i'} \sum_{A} x_{i, A} = n - 1 \enspace.
\]
Furthermore, we use feasibility of $x$ and the fact that $\lvert T \rvert \leq k$. This implies
\[
\sum_{i \neq i'} \sum_{A: A \cap T = \emptyset} x_{i, A} \leq \sum_{j \in T} \sum_{i \neq i'} \sum_{A: j \in A} x_{i, A} \leq \lvert T \rvert \leq k \enspace.
\]
Overall, this implies
\[
\sum_{i} W^{b_{-i}}(1-x_i) \geq (n - k - 1) \sum_{i} \sum_S b_{i}(S) \hat{x}_{i, S} = (n - k - 1) W^b(1) \enspace.
\]
As $W^b(1) \geq W^{b_{-i}}(1)$ for all $i$, this shows the claim.\qquad
\end{proof}

\begin{lemma}\label{lem:ca-smooth}
The pay-your-bid mechanism that solves the configuration LP for $\mathcal{MPH}-k$ valuations exactly is $(1/2,d+1)$-smooth for deviations to $b'_i = \frac{1}{2}v_i.$
\end{lemma}
\begin{proof}
Following the same steps as in the proof of Lemma~\ref{lem:pip-smooth} and using Lemma~\ref{lem:lemma4.7-mphk} instead of Lemma~\ref{lem:lemma-4.7-lps} completes the proof.\qquad
\end{proof}

\section{Extensions}
\label{sec:extensions}
Throughout this paper, we focused on pay-your-bid payment schemes. However, all of our results generalize to payment schemes that use arbitrary non-negative payments which are upper bounded by the respective bid. In this case, we resort to weak smoothness \cite{SyrgkanisT13}. In our statements $(\lambda, \mu)$-smoothness would be replaced by weak $(\lambda, 0, \mu)$-smoothness. Considering equilibria without overbidding, i.e., always $b_i(x) \leq v_i(x)$, this implies a PoA bound of $(1+\mu)/\lambda$.

Furthermore, Theorem~\ref{thm:main} also holds if $f'$ is not an exact declared welfare maximizer, but only allows implementation as a truthful mechanism. The interesting consequence is that it might make sense to only approximately solve the relaxation if this improves the smoothness guarantees. For example, a packing LP can be solved using the fractional-overselling mechanism in \cite{Hoefer2013}, which was originally introduced in \cite{Krysta2012}. The allocation rule is an $O(\log n + \log L)$-approximation for any packing LP with $n$ bidders and $L$ constraints between bidders. It allows implementation as a truthful mechanism but  it is also a greedy algorithm in the sense of \cite{LucierBorodin10}. Therefore, the respective pay-your-bid mechanism is $(\frac{1}{O(\log n + \log L)}, 1)$-smooth. This means that combining this algorithm with any $\alpha$-approximate oblivious rounding scheme for the respective packing LP, we get a pay-your-bid mechanism with Price of Anarchy at most $O(\alpha(\log n + \log L))$.

Finally, Carr and Vempala \cite{CarrVempala02} introduced randomized metarounding, which is a technique to derive oblivious rounding schemes from non-oblivious ones. Lavi and Swamy \cite{LaviS05} used this result to construct truthful mechanisms. However, they additionally need a packing structure. As in our case oblivious rounding is enough, any rounding scheme derived from the original version in \cite{CarrVempala02} is enough for our considerations.

\section{Discussion}
\label{sec:discussion}
In this paper we have shown that algorithms that follow the relax-and-round paradigm and 
whose rounding is oblivious have a very desirable property: Namely, if the rounding scheme is 
$\alpha$-approximate and the relaxation has a Price of Anarchy via smoothness of $O(\beta)$, 
then the resulting relax-and-round mechanism has a Price of Anarchy of $O(\alpha\beta)$
provable via smoothness.

Two aspects that we did not touch upon are equilibrium existence and the computational 
complexity of computing an equilibrium. The former is particularly relevant for pure equilibrium concepts,
such as pure Nash equilibria or pure Bayes-Nash equilibria. The latter has been shown to be a problem
for Bayes-Nash equilibria in simultaneous first-price auctions \cite{CaiP14}.

Our foremost intended application is to repeated settings, where regret minimization converges
to a coarse correlated equilibrium in polynomial-time. In fact, we do not even need vanishingly small 
regret---we only need that agents have no regret for deviations to half their value. This argument
readily applies to settings of incomplete information, showing near-optimal system performance even 
out of equilibrium. We consider the availability of such simple fall-back strategies as a major advantage of 
direct mechanisms over indirect mechanisms, where bidders typically have to solve a non-trivial
problem to figure out good strategies.

In terms of future work, it would be interesting to identify further algorithm design paradigms that
translate approximation guarantees into Price of Anarchy guarantees or, more generally, to obtain
a combinatorial characterization of algorithms with low price of anarchy. A first step towards this 
direction is \cite{DuttingK15}, which provides such a characterization for single-parameter settings.

\section*{Acknowledgements}
\'Eva Tardos is supported in part by NSF grants CCF-0910940 and CCF-1215994, ONR grant N00014-08-1-0031, a Yahoo! Research Alliance Grant, and a Google Research Grant.

\bibliographystyle{abbrvnat}
\bibliography{abb,approximation}

\appendix

\section{Improved Bound for Traveling Salesman Problem}
\label{app:tsp}
In this appendix we show how to improve the Price of Anarchy bound from $12$ to $9$ by using the algorithm of Paluch et al.~\citet{PaluchEZ12} instead.

The idea behind the algorithm of Paluch et al.~\citet{PaluchEZ12} is to compute a cycle cover that does not contain cycles of length $2$. However, computing such a cycle cover would be NP-hard, therefore Paluch et al.~relax the constraints further by allowing so-called half-edges. Informally spoken, in a cycle cover without $2$-cycles but with half-edges, each vertex still has in- and out-degree exactly $1$. The difference is that the orientation of one edges does not need to be consistent between its two endpoints. That is, the edge can can serve as the ingoing (or outgoing) for both endpoints. However, it is not possible that a single edge serves as both ingoing and outgoing edge for a vertex.

Formally, these constraints can be modeled as follows. A pair of directed edges $(u, v)$ and $(v, u)$ in the original graph is represented by an additional vertex $v_{u, v}$ and directed edges $(u, v_{u, v})$, $(v_{u, v}, v)$, $(v, v_{u, v})$, and $(v_{u, v}, u)$. In a cycle cover without $2$-cycles but with half-edges, either none or exactly two of these edges are included and every (original) vertex has indegree and outdegree exactly $1$.

Paluch et al.~show that a cycle cover without $2$-cycles but with half-edges can be computed in polynomial time. Furthermore, it is possible to derive three Hamiltonian cycles such that each half-edge is included in at least one of these cycles. Choosing one of these at random, we get an oblivious rounding scheme.

A similar argument as the one used to establish Lemma \ref{lem:cc-smooth} can be used to establish a bound for cycle covers without $2$-cycles but with half edges. Due to a single player now controlling two edges, however, the resulting bound can be as bad as $6$. So this would only yield a Price of Anarchy of $18$ via Theorem \ref{thm:smooth}. We therefore deviate from our proof pattern and show that the relaxation is $(1/2,3)$-smooth for deviations to $b'_i = \frac{1}{2}v_i$ if we only compare to a subset of the cycle covers without two-cycles but with half edges, namely the set of all tours. This turns out to be sufficient for Theorem~\ref{thm:smooth} to go through.

\begin{lemma}\label{lem:half-edges}
Consider bids $b$. Let $C$ be the cycle cover without 2-cycles but with half edges that maximizes reported welfare for bids $b$ and use $E_{C}$ to denote the set of edges used in this cycle cover. Consider any tour $T$ with edge set $E_{T}$. Then,
\[
	\sum_{i \in N: e_i \in E_{T}} \left(W^{b_{-i}}(\mathcal{C}) - W^{b_{-i}}(\mathcal{C}_{e_i}) \right) \le 3 \cdot W^b(\mathcal{C}).
\]
\end{lemma}
\begin{proof}
We begin by showing how to incorporate a given edge $e \in E$ from the tour into the cycle cover $C$ without $2$-cycles but with half edges by using as many edges from $E_C$ as possible.

Figure \ref{fig:half-edges} depicts the two base cases that can occur. The thin edges are edges from $E_C$. The thick edge is $e$.
Note that if we have $v_1 = v_4$ then we can w.l.o.g.~assume that $v_1 \neq v_6$. (As in the interesting case where $E_C$ contains $(v_3,v_{3,4})$ and $(v_4,v_{4,3})$ having $v_1 = v_6$ would mean that the node $v_1 = v_4 = v_6$ would only have one outgoing edge, which would contradict the fact that $E_C$ is a valid cycle cover without $2$-cycles but with half edges.) Hence the case $v_1 = v_4$ is covered by Case 1 below; and in Case $2$ we can assume that $v_1 \neq v_4.$

The first case is when nodes $v_1$ and $v_6$ are distinct and there is no edge between $v_1$ to $v_6$ in the cycle cover. In this case we can remove edges $(v_3,v_1)$ and $(v_6,v_4)$ from $E_C$ and add an edge between $v_1$ and $v_6$. We direct the half-edges in this edge so that they fit the in-/outgoing edges at $v_1$ and $v_6$.
The resulting edge set is a valid cycle cover without $2$-cycles but with half-edges because this modification to $E_C$ maintains the in-/outdegrees of all nodes, does not create self-loops (as we assumed that $v_1$ and $v_6$ are distinct), and does not create a $2$-cycle (as we assumed that the edge $v_1$ to $v_6$ is not contained in the cycle cover).

The second case is when either there is an edge from $v_1$ to $v_6$ in the cycle cover or the nodes $v_1$ and $v_6$ are identical. In this case the modification just described would either create a self-loop from/to node $v_1 = v_6$ or a $2$-cycle from $v_1$ to $v_6$ and back.
What we can do instead is remove edges $(v_3,v_1)$, $(v_6,v_4)$, and $(v_4,v_2)$ from $E_C$ and add edges between $v_4$ and $v_1$ and between $v_6$ and $v_2$. We direct the half-edges in these edges so that they fit with the in/outgoing edges at $v_1$, $v_2$ and $v_6$. This leads to a valid cycle cover without $2$-cycles but with half edges as it maintains in- and outdegrees and does not create self-loops or a $2$-cycle. (For the former note that in this case $v_1 \neq v_4$ and that $v_2 = v_6$ would imply the existence of a $2$-cycle. For the latter observe that we need not be worried about an edge between $v_1$ and $v_4$ if $v_1 = v_6$ as we remove the existing edge between these nodes. If $v_1 \neq v_6$ then such an edge cannot exist as it would mean $v_1$ had three incident edges. Similarly, an edge between $v_2$ and $v_6$ would mean $v_6$ had three incident edges.)

By the same argument as in Lemma~\ref{lem:cc-smooth} we obtain
\[
	\sum_{i \in N: e_i \in E_{T}} \left(W^{b_{-i}}(\mathcal{C}) - W^{b_{-i}}(\mathcal{C}_{e_i}) \right) \le 3 \cdot W^b(\mathcal{C}). \qedhere
\]
\end{proof}

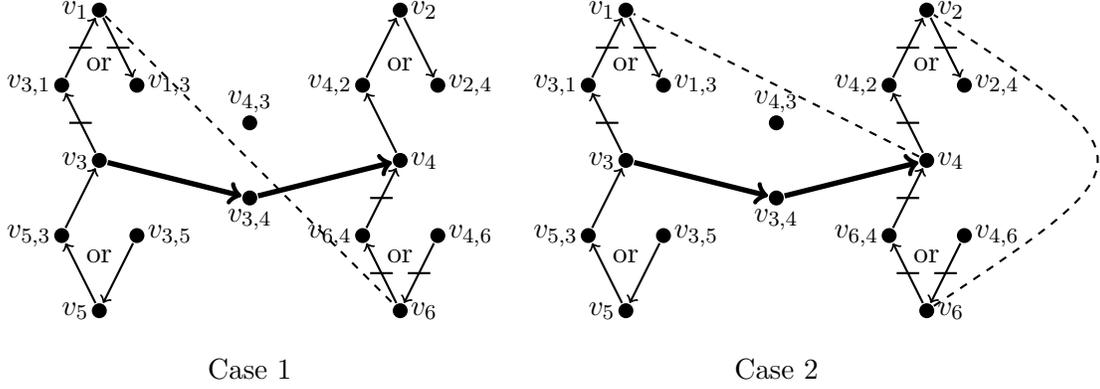
\begin{figure}
\center
\hspace{14pt}\mbox{\scalebox{1.0}{
\begin{tikzpicture}[thick]%
    \draw (2,-0.5) node[below](case2-label){Case 1};
    \draw (0,0) node[circle,fill,inner sep=2pt](v5){};
    \draw (0,0) node[left](v5-label){$v_5$};
    \draw (-0.5,1) node[circle,fill,inner sep=2pt](v53){};
    \draw (-0.5,1) node[left](v53-label){$v_{5,3}$};
    \draw (0.5,1) node[circle,fill,inner sep=2pt](v35){};
    \draw (0.5,1) node[right](v35-label){$v_{3,5}$};
    \draw (0,2) node[circle,fill,inner sep=2pt](v3){};
    \draw (0,2) node[left](v3-label){$v_3$};
    \draw (-0.5,3) node[circle,fill,inner sep=2pt](v31){};
    \draw (-0.5,3) node[left](v31-label){$v_{3,1}$};
    \draw (0.5,3) node[circle,fill,inner sep=2pt](v13){};
    \draw (0.5,3) node[right](v13-label){$v_{1,3}$};
    \draw (0,4) node[circle,fill,inner sep=2pt](v1){};
    \draw (0,4) node[left](v1-label){$v_1$};
    \draw (4,0) node[circle,fill,inner sep=2pt](v6){};
    \draw (4,0) node[right](v6-label){$v_6$};
    \draw (3.5,1) node[circle,fill,inner sep=2pt](v64){};
    \draw (3.5,1) node[left](v64-label){$v_{6,4}$};
    \draw (4.5,1) node[circle,fill,inner sep=2pt](v46){};
    \draw (4.5,1) node[right](v46-label){$v_{4,6}$};
    \draw (4,2) node[circle,fill,inner sep=2pt](v4){};
    \draw (4,2) node[right](v4-label){$v_4$};
    \draw (3.5,3) node[circle,fill,inner sep=2pt](v42){};
    \draw (3.5,3) node[left](v42-label){$v_{4,2}$};
    \draw (4.5,3) node[circle,fill,inner sep=2pt](v24){};
    \draw (4.5,3) node[right](v24-label){$v_{2,4}$};
    \draw (4,4) node[circle,fill,inner sep=2pt](v2){};
    \draw (4,4) node[right](v2-label){$v_2$};
    \draw (2,1.5) node[circle,fill,inner sep=2pt](v34){};
    \draw (2,1.5) node[below](v34-label){$v_{3,4}$};
    \draw (2,2.5) node[circle,fill,inner sep=2pt](v43){};
    \draw (2,2.5) node[above](v43-label){$v_{4,3}$};
    \draw[->] (v5) -- (v53);
    \draw[->] (v53) -- (v3);
    \draw[->] (v35) -- (v5);
    \draw[->] (v3) -- (v31);
    \draw[->] (v31) -- (v1);
    \draw[->] (v2) -- (v24);
    \draw[->] (v1) -- (v13);
    \draw (0,3.5) node[below](or){or};
    \draw (4,0.5) node[above](or){or};
    \draw (0,0.5) node[above](or){or};
    \draw (4,3.5) node[below](or){or};
    \draw[->] (v6) -- (v64);
    \draw[->] (v64) -- (v4);
    \draw[->] (v4) -- (v42);
    \draw[->] (v42) -- (v2);
    \draw[->] (v46) -- (v6);
    \draw[->, line width = 2pt] (v3) -- (v34);
    \draw[->, line width = 2pt] (v34) -- (v4);
    \draw[-, dashed] (v1) -- (v6);
    \draw (-0.4,3.5) -- (-0.1,3.5);
    \draw (0.1,3.5) -- (0.4,3.5);
    \draw (-0.4,2.5) -- (-0.1,2.5);
    \draw (3.6,0.5) -- (3.9,0.5);
    \draw (3.6,1.5) -- (3.9,1.5);
    \draw (4.1,0.5) -- (4.4,0.5);

    \draw (9,-0.5) node[below](case2-label){Case 2};
    \draw (7,0) node[circle,fill,inner sep=2pt](v5){};
    \draw (7,0) node[left](v5-label){$v_5$};
    \draw (6.5,1) node[circle,fill,inner sep=2pt](v53){};
    \draw (6.5,1) node[left](v53-label){$v_{5,3}$};
    \draw (7.5,1) node[circle,fill,inner sep=2pt](v35){};
    \draw (7.5,1) node[right](v35-label){$v_{3,5}$};
    \draw (7,2) node[circle,fill,inner sep=2pt](v3){};
    \draw (7,2) node[left](v3-label){$v_3$};
    \draw (6.5,3) node[circle,fill,inner sep=2pt](v31){};
    \draw (6.5,3) node[left](v31-label){$v_{3,1}$};
    \draw (7.5,3) node[circle,fill,inner sep=2pt](v13){};
    \draw (7.5,3) node[right](v13-label){$v_{1,3}$};
    \draw (7,4) node[circle,fill,inner sep=2pt](v1){};
    \draw (7,4) node[left](v1-label){$v_1$};
    \draw (11,0) node[circle,fill,inner sep=2pt](v6){};
    \draw (11,0) node[right](v6-label){$v_6$};
    \draw (10.5,1) node[circle,fill,inner sep=2pt](v64){};
    \draw (10.5,1) node[left](v64-label){$v_{6,4}$};
    \draw (11.5,1) node[circle,fill,inner sep=2pt](v46){};
    \draw (11.5,1) node[right](v46-label){$v_{4,6}$};
    \draw (11,2) node[circle,fill,inner sep=2pt](v4){};
    \draw (11,2) node[right](v4-label){$v_4$};
    \draw (10.5,3) node[circle,fill,inner sep=2pt](v42){};
    \draw (10.5,3) node[left](v42-label){$v_{4,2}$};
    \draw (11.5,3) node[circle,fill,inner sep=2pt](v24){};
    \draw (11.5,3) node[right](v24-label){$v_{2,4}$};
    \draw (11,4) node[circle,fill,inner sep=2pt](v2){};
    \draw (11,4) node[right](v2-label){$v_2$};
    \draw (9,1.5) node[circle,fill,inner sep=2pt](v34){};
    \draw (9,1.5) node[below](v34-label){$v_{3,4}$};
    \draw (9,2.5) node[circle,fill,inner sep=2pt](v43){};
    \draw (9,2.5) node[above](v43-label){$v_{4,3}$};
    \draw[->] (v5) -- (v53);
    \draw[->] (v53) -- (v3);
    \draw[->] (v35) -- (v5);
    \draw[->] (v3) -- (v31);
    \draw[->] (v31) -- (v1);
    \draw[->] (v2) -- (v24);
    \draw[->] (v1) -- (v13);
    \draw (7,3.5) node[below](or){or};
    \draw (11,0.5) node[above](or){or};
    \draw (7,0.5) node[above](or){or};
    \draw (11,3.5) node[below](or){or};
    \draw[->] (v6) -- (v64);
    \draw[->] (v64) -- (v4);
    \draw[->] (v4) -- (v42);
    \draw[->] (v42) -- (v2);
    \draw[->] (v46) -- (v6);
    \draw[->, line width = 2pt] (v3) -- (v34);
    \draw[->, line width = 2pt] (v34) -- (v4);
    \draw[-, dashed] (v4) -- (v1);
    \draw[-, dashed] (v6) .. controls(14,2) .. (v2);
    \draw (6.6,3.5) -- (6.9,3.5);
    \draw (7.1,3.5) -- (7.4,3.5);
    \draw (6.6,2.5) -- (6.9,2.5);
    \draw (10.6,0.5) -- (10.9,0.5);
    \draw (11.1,0.5) -- (11.4,0.5);
    \draw (10.6,1.5) -- (10.9,1.5);
    \draw (10.6,3.5) -- (10.9,3.5);
    \draw (11.1,3.5) -- (11.4,3.5);
    \draw (10.6,2.5) -- (10.9,2.5);
\end{tikzpicture}
}\hspace{-14pt}
}
\vspace{-5pt}
\caption{Smoothness of Cycle Cover LP without 2-Cycles but with Half Edges}\label{fig:half-edges}
\vspace{-5pt}
\end{figure}

\begin{lemma}
Consider valuation profile $v$ and bid profile $b$.  Denote the welfare achieved by the welfare maximizing tour by $OPT_T(v)$. Then for the pay-your-bid mechanism that computes an optimal cycle cover without two-cycles but with half edges and bids $b'_i = \frac{1}{2}v_i$ for all $i \in N$, 
\[
	\sum_{i \in N} u_i((b'-i,b_{-i}),v_i) \ge \frac{1}{2} OPT_T(v) - \mu \cdot \sum_{i \in N} p_i(b).
\]
\end{lemma}
\begin{proof}
The proof is analogous to the proof of Lemma~\ref{lem:pip-smooth}. The only two differences are that (1) instead of switching to the optimal relaxed solution for $v$ we switch to the optimal original solution for $v$ and (2) we then apply Lemma~\ref{lem:half-edges} instead of Lemma~\ref{lem:lemma-4.7-lps} to bound the social cost $\sum_{i \in N: e_i \in E_{T}} (W^{b_{-i}}(\mathcal{C}) - W^{b_{-i}}(\mathcal{C}_{e_i})$ of enforcing the optimal solution $T$.\qquad
\end{proof}

The only change to the proof of Theorem~\ref{thm:smooth} is that we do not need the extra step of lower bounding the optimal relaxed solution with the optimal original solution when we apply the smoothness guarantee of the mechanism for the relaxation.

\end{document}